\documentclass[11pt]{article}
\usepackage[utf8]{inputenc}
\usepackage{amsfonts, amsmath}
\usepackage{amssymb}
\usepackage{amsthm,amsmath,amscd}
\usepackage{enumerate}
\usepackage{url}
\usepackage[pdftex]{graphicx}
\usepackage[margin=25pt,font=small,labelfont=bf]{caption}
\usepackage{subfig}
\usepackage[numbers,square,sort&compress]{natbib}
\usepackage[colorlinks=true]{hyperref}

\newcommand*{\doi}[1]{doi: \href{https://dx.doi.org/#1}{\urlstyle{rm}\nolinkurl{#1}}}
\newcommand*{\arxiv}[1]{arXiv:  \href{https://arxiv.org/abs/#1}{\urlstyle{rm}\nolinkurl{#1}}}
\let\oldproofname=\proofname
\renewcommand{\proofname}{\rm\bf{\oldproofname}}

\newcommand{\RR}{\mathbb R}

\newcommand{\NN}{\mathbb N}
\newcommand{\ZZ}{\mathbb Z}

\newcommand{\bna}{\begin{eqnarray}}
\newcommand{\ena}{\end{eqnarray}}
\newcommand{\ba}{\begin{eqnarray*}}
\newcommand{\ea}{\end{eqnarray*}}
\newcommand{\bs}[1]{}

\newcommand{\f}{{\mathbf f}}

\newtheorem{theorem}{Theorem}[section]

\newtheorem{lemma}[theorem]{Lemma}

\newtheorem{remark}[theorem]{Remark}

\newcommand{\eps}{\varepsilon}

\def\e{{\bf e}}

\def\0{{\bf 0}}
\def\b{{\bf b}}

\let\oldv=\v
\def\v{{\bf v}}

\def\e{{\bf e}}

\def\x{{\bf x}}

\usepackage[margin=1in]{geometry}
\begin{document}
\title{Lattices Without a Big Constant and With Noise}

\author{
Steven J. Gortler and 
Louis Theran
}
\date{}
\maketitle 

\begin{abstract}
We show how Frieze's analysis of subset sum solving using lattices can
be done with out any large constants 
and without flipping.  We apply the variant without the large constant 
to inputs with noise.
\end{abstract}

\section{Introduction}
In~\cite{LO}, 
Lagarias and Odlyzko 
introduce a lattice based method for 
efficiently solving subset sum problems
with large integers, with high probability.
In~\cite{frieze}, Frieze applies two slight alterations to the method of 
Lagarias and Odlyzko, and
is then able to provide a very simple proof of the high probability correctness
of  the altered method. The first alteration, as we describe below is the 
introduction of a large constant in their lattice construction. The
second introduction is testing a certain condition on the input and 
if this condition fails, instead solving a suitably ``flipped'' problem.

In this note we first show that one can avoid the use of this large constant without
sacrificing the result or making the proof (much) more complicated. We next show
that if we alter the problem by adding one extra row to the lattice construction,
we can avoid the test-and-flip step.

Our motivation for removing the large constant comes from our desire to 
deal with slightly noisy input. In \cite{gam,gam2}, the authors further altered 
Frieze's lattice contstruction to deal with noisy input. In this note
we also show, that once the large constant is removed, small noise can
be dealt with without any alterations at all.

The ideas in this note have been applied in~\cite{CGLline} where we use
a lattice based approach to reconstruct a one-dimensional point configuration
from an unlabled subset of the interpoint distances.

This text will be borrow quite heavily from Frieze's language verbatim 
throughout, without further specific attribution, 
and we will
assume the reader is quite familiar with that paper.

\section{Basic Lattice Construction}

Let $\e= [e_1; e_2;\ldots ;e_n]\in\{0,1\}^n$ be fixed.
Let $B_1, B_2,\ldots B_n$ be postive integers and
$B_0=\sum_{i=1}^nB_ie_i$. 
We will assume that $B_1, B_2,\ldots B_n$ are
independently chosen at random from $1,\ldots,B=2^{cn^2}$
with $c=1/2+\epsilon$ for $\epsilon>0$.
The given SUBSET-SUM problem is to find $\e$ given 
$B_0, B_1, ...B_n$.

Adapting the method of~\cite{LO}, Frieze 
assumes that
\bna 
\label{eq:ass}
B_0 \ge \frac{1}{2}\sum_{i=1}^n B_i
\ena 
If this doesn't hold, he instead replaces $B_0$ by
\ba 
    \left(\sum_{i=1}^n B_i\right) - B_0
\ea 
Plainly, this new problem has a solution iff the original does, and one 
can easily determine $\e$ from the solution to the new problem.
We call this step ``test-and-flip''.

Adapting the method of~\cite{LO}, Frieze then defines 
the integer 
lattice that is generated by the 
($n+1$)-by($n+1$) matrix:
\bna
\label{eq:original}
\begin{pmatrix}
pB_0 & -pB_1 & -pB_2 &... & -pB_{n-1} &-pB_n\\
0 & 1 & 0& ... & 0 &0\\
0 & 0 & 1& ... & 0 &0\\
... & ... & ...&... & ... &...\\
0 & 0 & 0&... & 1 &0\\
0 & 0 & 0& ... & 0 &1\\
\end{pmatrix}
\ena
where the integer $p>n2^{n/2}$ plays the role of a ``large constant''.

The LLL algorithm~\cite{LLL}, (an efficient algorithm) 
is then run to find a small
vector in this lattice.

We will state the main result of Frieze in the following form:
\begin{theorem}\label{thm:main}
For all $\epsilon > 0$, there is an $n_0(\epsilon) \in \NN$,
such that if $n>n_0$,
the algorithm returns a vector that is a scale factor of
$\e$
with probability at least $1 - 2^{-\epsilon n^2/2}$.
\end{theorem}

(The use of $n_0$ mirrors the use in~\cite{LO}
immediately following their Theorem 3.5.)

\section{Removing the large constant}
Here we show that the large constant can be removed without
changing the correctness of the method, and without changing
the proof too much.

Without the large constant,
similarly to~\cite{LO}, the lattice $L$ will be generated by the 
columns of the ($n+1$)-by($n+1$) matrix:
\ba
\begin{pmatrix}
B_0 & -B_1 & -B_2 &... & -B_{n-1} &-B_n\\
0 & 1 & 0& ... & 0 &0\\
0 & 0 & 1& ... & 0 &0\\
... & ... & ...&... & ... &...\\
0 & 0 & 0&... & 1 &0\\
0 & 0 & 0& ... & 0 &1\\
\end{pmatrix}
\ea
with columns $\b_0....\b_n$.

The LLL algorithm is guaranteed to 
find us $ {\x}\in L$, $ {\x}\neq 0$ satisfying
\ba
|| {\x}|| \le  2^{n/2} ||\e||\le  2^{n/2}n^{1/2} =:m
\ea

If $ {\x} = [x_0; x_1; ...x_n]\in L$ then we have
\ba
 {\x} = x'_0\b_0 + x_1\b_1 +...+x_n\b_n
\ea
where
\ba
x_0 = \left( B_0x'_0-\sum_{i=1}^{n}B_ix_i\right)
\ea

Let $A = \{ {\x} \in L, ||\x||\le m, \x \neq k [0;{\e}]\}$
for any $k \in \ZZ$. 
{\bf This is the set that can give rise to algorithmic failure.}

But if $ {\x} \in A$ then 
\ba
|B_0 x'_0|=\left|x_0+\sum_{i=1}^nB_ix_i  \right|\le 
|x_0|+\sum_{i=1}^n B_i|| {\x}||
\ea
Using (\ref{eq:ass}), this gives us
\ba
|x'_0|
&\le& 
\frac{|x_0|}{B_0}+\frac{\sum_{i=1}^n B_i}{B_0}||\x||\\
&\le& 
\frac{|x_0|}{B_0}+2|| {\x}||\\
&\le& 
3m
\ea
Note that we get a $3m$ instead of Frieze's $2m$ but this will not be material, as we shall see below.

So if $A \neq \emptyset$ there exists 
$ {\x} = [x_0; x_1; x_2; ...;x_n] \in \ZZ^{n+1}$ and $y\in\ZZ$ satisfying
\begin{align}|| {\x}|| < m, \;\;\; |y|\le 3m \tag{a} \\
 {\x} \neq k[0; \e] \;\;\; {\rm for\, any\,} k\in\ZZ \tag{b}\\
\sum_{i=1}^nB_ix_i = yB_0 - x_0 \tag{c}
\end{align}

Consider now a fixed $( {\x},y)$ satisfying (a) and (b), we will prove that
\ba
\operatorname{Pr}( {\x},y {\rm \,satisfy\, (c)}) \le 1/B
\ea

To prove this, note that (c) is equivalent to $\sum_{i=1}^n B_i z_i = -x_0$
where $z_i=x_i-ye_i$.
This is simply a non-trivial (due to (b)) inhomogeneous linear equation over the $\{B_1, ..., B_n\}$.

\begin{lemma}
\label{lem:aff}
Let $H$ be a
$d$-dimensional affine subset of $\RR^{n}$. 
The number of points in the discrete cube 
$[1...B]^n$ 
intersected with $H$ is at most $B^d$. 
\end{lemma}
\begin{proof}
Let $C^n$ be the discrete cube of the statement.
Define $C^d$ similarly.
The projection of $H$ onto the 
the first $d$ coordinates by forgetting the last $n-d$ coordinates 
is  a bijective (if not, pick a different coordinate subspace), and in particular 
injective,
affine map that sends  points in $H\cap C^n$ to 
points $C^d$.
It follows that $|H\cap C^n| \le |C^d| = B^d$.  
\end{proof}

Since we have one equation, we get $d=n-1$ in our application of Lemma~\ref{lem:aff}, 
giving us $\frac{B^{n-1}}{B^n}=1/B$.

\begin{remark}
In Frieze's original method, he gets $x_0=0$, and so his linear equation is guaranteed to be homogeneous. This does not effect
the count of Lemma~\ref{lem:aff}. Frieze states his argument for this step probabilistically, but we prefer the more general linear algebraic
interpretation. Note the Frieze's probabilistic argument could have worked in the inhomogeneous case as well.
\end{remark}

Letting 
$A_1=\{ {\x} \in \ZZ^{n+1} : ||\x||\le m\}$.
and summing over all $A_1$ and $y$, we get a failure probability bound of 
\bna
\label{eq:count}
\frac{(6n+1)|A_1|}{B} \le 
\frac{(6n+1)(2m+1)^{n+1}}{B} \le \frac{2^{n^2/2}2^{O(n\log n)}}{B} \le O(2^{-\epsilon n^2/2})
\ena
We get $(2m+1)^{n+1}$ instead of  $(2m+1)^n$ as in \cite{frieze}, but this is subsumed 
into the $2^{O(n\log n)}$ in the next step.
We expand on the last inequality in the following lemma.

\begin{lemma}\label{lem: failprob}
For fixed $\epsilon > 0$.
For $n$ sufficiently large (depending on $\epsilon$), 
the probability of failure
is at most
$\frac{1}{2^{\epsilon n^2 /2}}$
\end{lemma}
\begin{proof}
Recall 
$B=2^{cn^2}$
where  $c=1/2+\epsilon$ for $\epsilon>0$.
For sufficiently large $n$, the quantity $2^{n^2/2}2^{O(n\log n)}$
is bounded by $2^{(\frac{1}{2} + \frac{\epsilon}{2})n^2}$.  The 
probability of the event in the statement 
is then at 
most 
\[
    \frac{2^{(\frac{1}{2} + \frac{\epsilon}{2})n^2}}{2^{(\frac{1}{2}+\epsilon)n^2}} = \frac{1}{2^{\epsilon n^2 /2}}
\]
\end{proof}

This Lemma then establishes the result of Theorem~\ref{thm:main}.

In summary, Frieze's argument goes through without the use of the large constant. 
The bound on $y$ becomes slightly worse, as does a term in the count of Equation~\ref{eq:count}, 
but all of this is swallowed up by the $2^{O(n\log n)}$ term. Perhaps the biggest difference is that
his original homogeneous linear constraint becomes inhomogeneous, but both cases are covered by
Lemma~\ref{lem:aff}.

\section{Removing the flip}

Here we show that the assumption of Equation~\ref{eq:ass} can be removed if we add one more row to
our lattice generating matrix.

We will leave the big constant out.

So now our lattice $L$ will be generated by the 
columns of the ($n+2$)-by-($n+1$) matrix:
\ba
\begin{pmatrix}
B_0 & -B_1 & -B_2 &... & -B_{n-1} &-B_n\\
1 & 0 & 0& ... & 0 &0\\
0 & 1 & 0& ... & 0 &0\\
0 & 0 & 1& ... & 0 &0\\
... & ... & ...&... & ... &...\\
0 & 0 & 0&... & 1 &0\\
0 & 0 & 0& ... & 0 &1\\
\end{pmatrix}
\ea
Note the new second row.

The LLL algorithm will find us $ {\x}\in L$, $ {\x}\neq 0$ satisfying
\ba
|| {\x}|| \le  2^{(n+1)/2} ||[1;\e]||\le  2^{(n+1)/2}(n+1)^{1/2} 
\le   2\cdot 2^{n/2}n^{1/2}    =:m
\ea
This $m$ is slightly larger than that of the previous section, but not
in a way that will prove material.

If $ {\x} = [x_0; x_0'; x_1; ...x_n]\in L$ then we have
\ba
 {\x} = x'_0\b_0 + x_1\b_1 +...+x_n\b_n
\ea
where
\ba
x_0 = \left( B_0x'_0-\sum_{i=1}^{n}B_ix_i\right)
\ea

The main idea here is that the $x'_0$ data will show up in the lattice vector.
We will bounding the size of the lattice vectors, so we will not need any
extra bounding for $x'_0$, hence no need for a flip.

Let $A = \{ {\x} \in L, ||\x||\le m, \x \neq k [0;1;{\e}]\}$
for any $k \in \ZZ$.

So if $A \neq \emptyset$ there exists 
$ {\x} = [x_0; x'_0; x_1; x_2; ...;x_n] \in \ZZ^{n+2}$ satisfying
\begin{align} 
||\x|| < m,  \tag{a$'$}\\
\x \neq k[0; 1; \e] \;\;\; {\rm for\, any\,} k\in\ZZ \tag{b$'$}\\
\sum_{i=1}^nB_ix_i = x'_0 B_0  - x_0 \tag{c$'$}
\end{align}

Consider now a fixed $\x$ satisfying (a$'$) and (b$'$), we prove,
as above, that
\ba
\operatorname{Pr}(\x{\rm \,satisfy\, (c')}) \le 1/B
\ea

To prove this, note that (c$'$) is equivalent to 
$\sum_{i=1}^n B_i z_i = -x_0$
where $z_i=x_i-x'_0e_i$.
Again, this is an inhomogeneous linear equation
over $\{B_1...B_n\}$.

Letting 
$A_1=\{ {\x} \in \ZZ^{n+2} : ||\x||\le m\}$.
and summing over all $A_1$, we get 
\ba
\frac{|A_1|}{B} \le \frac{(2m+1)^{n+2}}{B} 
\le \frac{2^{n^2/2}2^{O(n\log n)}}{B} \le O(2^{-\epsilon n^2/2})
\ea
Note that we get an $n+2$ exponent in the second term instead of Frieze's $n$, but this is not
material for the third term.

From what we have gleamed from~\cite{LO}
(which has its own version of flipping),
we suspect that the test-and-flip step can
be omitted, and no extra row needs to be added, 
without impacting the success of the algorithm.
(This is also consistent with our experiments, below.)
But proving this might 
require going back to the proof methods
of~\cite{LO}, which are more involved.

\section{Adding Noise}

Let $\eps$ be a $\{-1,0,1\}^n$ be a fixed noise vector. 
And suppose that instead of the correct $B_i$, 
we are given $\{B_0, B_1+\eps_1, ... B_2+\eps_n\}$. ($B_0$ itself is given 
without noise.)
We will see that whp, we can still solve the underlying
subset sum problem.

Our lattice will now be generated by the 
columns of the matrix
\ba
M:=
\begin{pmatrix}
B_0 & -B_1-\eps_1 & -B_2-\eps_2 &... & -B_{n-1}-\eps_{n-1} &-B_n -\eps_n\\
1 & 0 & 0& ... & 0 &0\\
0 & 1 & 0& ... & 0 &0\\
0 & 0 & 1& ... & 0 &0\\
... & ... & ...&... & ... &...\\
0 & 0 & 0&... & 1 &0\\
0 & 0 & 0& ... & 0 &1\\
\end{pmatrix}
\ea
Consider the vector 
$\f := M[1;\e]$ 
in the lattice.
We have $\f = [f_0;1;\e]$,
Because $\e$ solved the subset sum problem, 
we have an error term
$|f_0| \le n$ and thus 
$||\f|| \le (2n+1)^{1/2}$.

The LLL algorithm will find us $ {\x}\in L$, $ {\x}\neq 0$ satisfying
\ba
|| {\x}|| \le  2^{(n+1)/2} (2n+1)^{1/2}
\le   3\cdot 2^{n/2}n^{1/2}    =:m
\ea

If $ {\x} = [x_0; x_0'; x_1; ...x_n]\in L$ then we have
\ba
 {\x} = x'_0\b_0 + x_1\b_1 +...+x_n\b_n
\ea
where
\ba
x_0 = \left( B_0x'_0-\sum_{i=1}^{n}B_ix_i 
-\sum_{i=1}^n \eps_i x_i \right)
\ea

Let $A = \{ {\x} \in L, ||\x||\le m, \x \neq k 
[f_0;1;{\e}]\}$
for any $k \in \ZZ$.

So if $A \neq \emptyset$ there exists 
$ {\x} = [x_0; x'_0; x_1; x_2; ...;x_n] \in \ZZ^{n+2}$ satisfying
\begin{align}
    ||\x|| < m,  \tag{a$''$}\\
    \x \neq k[f_0; 1; \e] \;\;\; {\rm for\, any\,} k\in\ZZ \tag{b$''$}\\
    \sum_{i=1}^nB_ix_i = x'_0 B_0  - x_0 -\sum_{i=1}^n \eps_i x_i \tag{c$''$}
\end{align}

Consider now a fixed $\x$ satisfying (a$''$) and (b$''$).  We prove,
as above, that
\ba
\operatorname{Pr}(\x{\rm \,satisfies\, (c'')}) \le 1/B
\ea

To prove this, note that (c$''$) is equivalent to 
$\sum_{i=1}^n B_i z_i = -x_0  - E$
where $z_i=x_i-x'_0e_i$
and $E=\sum_{i=1}^n \eps_i x_i$
This is simply a non-trivial (due to (b$''$)) inhomogeneous linear equation over the $\{B_1, ..., B_n\}$.

Letting 
$A_1=\{ {\x} \in \ZZ^{n+2} : ||\x||\le m\}$.
and summing over all $A_1$, and all possible 
error vectors
we get 
\ba
\frac{3^n|A_1|}{B} \le \frac{3^n(2m+1)^{n+2}}{B} 
\le \frac{2^{n^2/2}2^{O(n\log n)}}{B} \le O(2^{-\epsilon n^2/2})
\ea
Again, our second term is larger than that of Frieze, but not materially so.

In this section on input errors we have also used the lattice where an extra second row was added to the lattice generating matrix
as in the previous section.
This second row could have been omitted as long we applied Frieze's test-and-flip step. The correctness of this
is left as an exercise.

\section{Experiments}

We implemented methods described above and experimented with random inputs.
We investigated a number of different variants:
The original Frieze method with test-to-flip, 
the original Frieze method with no test-to-flip
and our method with the introduced second row.
For all of these three methods, we tried with, and without, the 
introduction 
of the  large constant. In total this gave us $6$ methods to work with.
We have found that when using the theoretically 
prescribed input magnitude size, $B = 2^{n^2 / 2}$,
all $6$ methods work without any detected failures. 
Indeed, all $6$ methods 
continued to work  when using much smaller values for
$B$.

To push the methods,
we kept dropping $B$ until we came near a phase transition, where failures
began to appear. We found that, at this point, 
the $6$ different methods all failed
randomly, but there was no discernible difference in their
success rates.

We then experimented with the introduction of $\{-1,0,1\}$ noise
in the input. Here, we found that all three methods worked as 
expected as long as no large constant was used. When the large constant was
introduced, all three methods failed consistently.

\def\v{\oldv}



\begin{thebibliography}{6}
    \providecommand{\natexlab}[1]{#1}
    \providecommand{\url}[1]{\texttt{#1}}
    \expandafter\ifx\csname urlstyle\endcsname\relax
      \providecommand{\doi}[1]{doi: #1}\else
      \providecommand{\doi}{doi: \begingroup \urlstyle{rm}\Url}\fi
    
    \bibitem[Connelly et~al.(2020)Connelly, Gortler, and Theran]{CGLline}
    R.~Connelly, S.~J. Gortler, and L.~Theran.
    \newblock Reconstruction in one dimension from unlabeled euclidean lengths.
    \newblock \emph{arXiv preprint arXiv:2007.06550}, 2020.
    
    \bibitem[Frieze(1986)]{frieze}
    A.~M. Frieze.
    \newblock On the {L}agarias-{O}dlyzko algorithm for the subset sum problem.
    \newblock \emph{SIAM J. Comput.}, 15\penalty0 (2):\penalty0 536--539, 1986.
    \newblock \doi{10.1137/0215038}.
    
    \bibitem[Gamarnik et~al.(2019)Gamarnik, Kızıldağ, and Zadik]{gam2}
    D.~Gamarnik, E.~C. Kızıldağ, and I.~Zadik.
    \newblock Inference in high-dimensional linear regression via lattice basis
      reduction and integer relation detection.
    \newblock \emph{arXiv preprint arXiv:1910.10890}, 2019.
    
    \bibitem[Lagarias and Odlyzko(1985)]{LO}
    J.~C. Lagarias and A.~M. Odlyzko.
    \newblock Solving low-density subset sum problems.
    \newblock \emph{J. Assoc. Comput. Mach.}, 32\penalty0 (1):\penalty0 229--246,
      1985.
    \newblock \doi{10.1145/2455.2461}.
    
    \bibitem[Lenstra et~al.(1982)Lenstra, Lenstra, and Lov\'{a}sz]{LLL}
    A.~K. Lenstra, H.~W. Lenstra, Jr., and L.~Lov\'{a}sz.
    \newblock Factoring polynomials with rational coefficients.
    \newblock \emph{Math. Ann.}, 261\penalty0 (4):\penalty0 515--534, 1982.
    \newblock \doi{10.1007/BF01457454}.
    
    \bibitem[Zadik and Gamarnik(2018)]{gam}
    I.~Zadik and D.~Gamarnik.
    \newblock High dimensional linear regression using lattice basis reduction.
    \newblock In \emph{Advances in Neural Information Processing Systems 31}, pages
      1842--1852. 2018.
    
    \end{thebibliography}
\end{document}